\documentclass[titlepage,10pt]{article}
\usepackage{tocloft}
\usepackage[titletoc,title]{appendix}
\usepackage{setspace}
\usepackage{graphicx,color}
\usepackage{amsmath, amsthm, amssymb}
\usepackage{mathrsfs}
\usepackage{natbib}
\usepackage[breaklinks]{hyperref}

\usepackage{verbatim}
\usepackage{enumerate}
\usepackage[left=1.5in,top=1in,right=1in,bottom=1.5in]{geometry}

\newcommand{\bidmark}{\rm u}
\newcommand{\inmark}{\rm i}
\newcommand{\outmark}{\rm o}

\newcommand{\kin}{\ensuremath{{k^\text{(i)}}}}
\newcommand{\kout}{\ensuremath{{k^\text{(o)}}}}
\newcommand{\kbid}{\ensuremath{{k^\text{(u)}}}}
\newcommand{\kini}{\ensuremath{{k_i^\text{(i)}}}}
\newcommand{\kouti}{\ensuremath{{k_i^\text{(o)}}}}
\newcommand{\kbidi}{\ensuremath{{k_i^\text{(u)}}}}

\newcommand{\kvec}{\ensuremath{\mathbf{k}}}
\newcommand{\xvec}{\ensuremath{\mathbf{x}}}

\newcommand{\zeroone}{\ensuremath{ \{0,1\} } }

\newcommand{\onthresholdi}{\phi_{{\rm on}, i}}
\newcommand{\offthresholdi}{\phi_{{\rm off}, i}}
\newcommand{\onthreshold}{\phi_{{\rm on}}}
\newcommand{\offthreshold}{\phi_{{\rm off}}}
\newcommand{\avgdeg}{k_{\rm avg}}
\numberwithin{equation}{section}

\newtheorem{mythm}{Lemma}

\begin{document}

\pagestyle{empty}

\begin{titlepage}
  \begin{center}
    \vspace*{1.5in}
    \Large{ON-OFF THRESHOLD MODELS OF\\
      SOCIAL CONTAGION}
    
    \normalsize{
      \vspace{9em}
      A Thesis Presented\\[1em]
      by\\[1em]
      Kameron Decker Harris\\[1em]
      to\\[1em]
      The Faculty of the Graduate College\\[1em]
      of\\[1em]
      The University of Vermont\\[1em]
      In Partial Fulfillment of the Requirements\\
      for the Degree of Master of Science\\
      Specializing in Mathematics
      \vspace{1em}

      October, 2012
    }
  \end{center}
\end{titlepage}

\noindent
Accepted by the Faculty of the Graduate College, 
The University of Vermont, in partial
fulfillment of the requirements for the degree of 
Master of Science, specializing in Mathematics.

\vspace{3em}

\noindent Thesis Examination Committee:

\vspace{6em}

\line(1,0){200} \hspace{1em} Advisor

Peter Sheridan Dodds, Ph.D.

\vspace{6em}

\line(1,0){200}

Christopher M. Danforth, Ph.D.

\vspace{6em}

\line(1,0){200} \hspace{1em} Chairperson

Joshua Bongard, Ph.D.

\vspace{6em}

\line(1,0){200} \hspace{1em} Dean, Graduate College

Domenico Grasso, Ph.D.

\vfill
August 15, 2012

\pagestyle{empty}
\newpage
{\flushleft \Large\bf Abstract}
\vspace{1em}

\noindent
We study binary state contagion dynamics on a social network
where nodes act in response to the average state of their neighborhood.
We model the competing tendencies of imitation
and non-conformity by incorporating an off-threshold
into standard threshold models of behavior.
In this way, we attempt to capture important aspects of fashions and 
general societal trends.
Allowing varying amounts of stochasticity in both the network and
node responses, we find different outcomes in the random and
deterministic versions of the model. 
In the limit of a large, dense network, however,
we show that these dynamics coincide.
The dynamical behavior of the system ranges from 
steady state to chaotic depending on network connectivity and 
update synchronicity.
We construct a mean field theory for general random networks.
In the undirected case, the mean field theory predicts that the 
dynamics on the network are a smoothed version of the average
node response dynamics.
We compare our theory to extensive simulations on Poisson random graphs
with node responses that
average to the chaotic tent map.


\newpage
\pagenumbering{roman}
\setcounter{page}{2}
\pagestyle{plain}
\renewcommand{\contentsname}{Table of Contents}
\tableofcontents

\newpage
\addcontentsline{toc}{section}{\listtablename} 
\listoftables

\newpage
\addcontentsline{toc}{section}{\listfigurename}
\listoffigures

\newpage
\pagenumbering{arabic}
\setcounter{page}{1}
\pagestyle{plain}
\begin{doublespace}
  \section{Introduction}

Almost universally, people enjoy
observing, speculating, and arguing about 
their fellow humans' behavior, including what they are wearing,
the music they listen to, 
and however else they express themselves.
Fashions,  trends, and fads intrigue us, 
whether past, present, or future.
Contemporary examples of prevalent trends
include skinny jeans, fixed-wheel
bicycles, and products made by Apple Inc. 
In the 1960's,
shag carpets and floral wallpaper would take their place.
Trends are also prevalent in language:
``hipster'' was popular in the 1940's and evolved into ``hippie'',
but it has been reappropriated today with different connotations; 
``groovy'' was once used seriously, and now comes out tongue-in-cheek.

The dynamics of these cultural phenomena are fascinating and complex.
They reflect numerous factors such as the political climate,
social norms, technology, marketing, and history. 
These phenomena are also influenced
by essentially random events. 
Usually, we can explain the emergence of a fad after the fact, 
but it is extremely difficult
to predict {\it a priori} whether some behavior will become popular.
Furthermore, those behaviors which are adopted by a large fraction of the 
population can lose their excitement and die out forever or later
recur unpredictably.

In this work we describe a mathematical model of the rise and
fall of trends. 
In particular, we model a social contagion process
where people are influenced by the behavior of their friends.
The agents in the model act according to simple competing tendencies
of imitation and non-conformity. One can argue that these
two ingredients are essential to all trends; indeed,
Simmel, in his classic essay ``Fashion'' \citeyearpar{simmel1957a},
believed that these are the essential forces 
behind the creation and destruction of fashions.

Our model is not meant to be quantitative, except
perhaps in carefully designed experiments,
but it captures the features with which we are familiar: 
some trends take off and some do not, and
some trends are stable while others vary wildly through time. 
Our model is closely related to the seminal work
of \citet{schelling1971a} and \citet{granovetter1978a}.
Being a mathematical model, it is also connected to 
theories of percolation \citep{stauffer1994a}, 
disease spreading \citep{newman2003a}, 
and magnetism \citep{newman2003a, aldana2003a}.

We focus on the derivation and analysis
of dynamical master equations that describe the expected evolution
of the system state. 
Through these equations we are able to 
predict the behavior of the social contagion process. 
In some cases, this can be done by hand, 
but most of the time we resort to
numerical methods for their iteration or solution.

This thesis is structured as follows.
In Section~\ref{sec:background}, we introduce the reader to
important background material relevant to the on-off threshold model.
In Section~\ref{sec:model}, we define the model and its
deterministic and stochastic variants. 
In Section~\ref{sec:fixedanalysis},
we provide an analysis of the model when the underlying network is fixed.
Section~\ref{sec:meanfield} develops a mean field theory
of the model in the most general kind of random graphs.
In Section~\ref{sec:tentmapf}, we consider the model
on Poisson random graphs with a specific kind of response function.
Our analysis is then applied to this specific case, and we compare
the results of simulations and theory. 
Finally, Section~\ref{sec:conclusions} presents conclusions
and directions for further research.

\section{Background}

\label{sec:background}

\subsection{Basics of graphs and networks}

When modeling any dynamical system of many interacting particles or
agents, we are often forced to start with a simplified description
of their interactions. In a solid, for example, atoms are situated
on some sort of lattice and assumed to interact only with
their nearest neighbors.
However, in many cases the interactions we aim to model do not
have the periodic structure of a lattice. 
Graphs, which are just a set of points connected by lines,
are a more general mathematical structure.
We denote a graph $\mathscr{G}$ as an ordered pair
$\mathscr{G} = (V,E)$ where $V$ is the set of
vertices (also called nodes)
and $E \subseteq V \times V$ the set of edges (also called links).
Here, $E$ is a set of ordered pairs, and we 
denote an edge from vertices $i$ to $j$ as $ij \in E$.
The terms {\it graph} and {\it network} will
be used interchangeably, but they have acquired different connotations. 
\citet[\S 1.2.1]{hackett2011a} puts this
rather nicely:
\begin{quote}
\begin{singlespace}
  To model a complex system as a graph is to filter out the
  functional details of each of its components, and the idiosyncrasies of
  their interactions with each other, and to focus instead on the underlying
  structure (topology) as an inert mathematical construct. Although this
  technique is central also to network theory, the word network, in contrast,
  usually carries with it connotations of the context in which the overarching
  system exists, particularly when that system displays any sort of nonlinear
  dynamics. For example, when investigating the spread of infectious disease
  on a human sexual contact network it makes sense to consider the relevant
  sociological parameters as well as the abstract topology, and it is in such
  settings that the interdisciplinary aspect that distinguishes network theory
  comes to the fore.
\end{singlespace}
\end{quote}

In models of human behavior, interactions can be considered to
occur on a social network.
Each person is connected to those people they interact with.
Interactions that constitute a connection between two people,
A and B, may be defined in many ways.
For example: 
\begin{enumerate}[\quad{Case} 1.]
\item A contacts B. This could mean:
  \begin{enumerate}
  \item \label{link_dir} A sends B an email.
  \item \label{link_bi} A attends the same concert as B.
  \end{enumerate}
\item A is B's superior in a hierarchy. \label{link_tree}
\item A and B both belong to the same group.
\end{enumerate}
These examples illustrate different types of edges that can be 
included in a graph. In Case~\ref{link_dir}, sending an email message 
is a one-way communication, so it
forms a {\it directed} edge. 
We represent directed edges with arrows
in a drawing of the corresponding graph. 
In Case~\ref{link_bi}, however, attending a concert is a
symmetric relation and there is no reason to give the edge direction.
Such edges are called {\it undirected} or bidirectional.
See Figure~\ref{fig:small-graph} for an example of a small
graph containing multiple types of edges.
A simple graph, by definition, contains no directed edges or self
loops (edges connecting a vertex to itself). 
It will be useful to represent the connectivity of the graph with
the {\it adjacency matrix}, 
$A=(A_{ij})$, where $A_{ij} = 1$
if and only if $ji \in E$ 
(this ``backwards'' definition is not standard, but it is
useful for the linear algebra to come).
The unfamiliar reader is directed to \citet{west2001a}
for a thorough introduction to graph theory.

\begin{figure}
  \centering
  \includegraphics[width=0.5\linewidth]{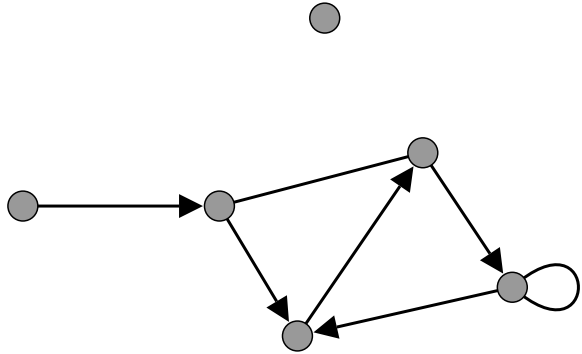}
  \caption[A graph consisting of
    a connected component of 5 nodes
    and one isolated node.]
  {A graph consisting of
    a connected component of 5 nodes
    and one isolated node. 
    The majority of the edges are directed.
    There is one undirected edge, shown without arrows.
    There is one self-loop which
    should be interpreted as directed.}
  \label{fig:small-graph}
\end{figure}

When we wish to analyze the structure of a given network,
one of the first things to examine is its degree distribution.
The {\it degree} $k_i$ of a vertex $i \in V$ is the number of 
edges incident to it, regardless of direction.
The sequence 
$k_1, k_2, \ldots, k_N$ 
is called the degree sequence. 
In simple graphs, this is just the size of a vertex's neighborhood.
In more complicated graphs we usually speak instead of a vertex $i$'s
in-degree $\kini$, out-degree $\kouti$,
and undirected degree $\kbidi$, all defined in the obvious way
\citep{dodds2011a}.
The {\it degree distribution} $p_k$
is a probability distribution
which tells us how edges are distributed among nodes.
The average degree $\avgdeg = \sum_{i=0}^{\infty} k p_k$ characterizes
the overall density of edges. A degree distribution which is
sharply peaked about its mean indicates a relatively homogeneous
network where vertices tend to have the same number of incident edges.
In contrast, a skewed distribution could result in some nodes having 
very high degree while the majority have low degree. This is common
in real networks, and the tail of the degree
distribution often follows an approximate
power law $p_k \sim k^{-\alpha}$ 
for some exponent $\alpha$ \citep{newman2003a}. 

\subsection{Random graphs}

Random graphs are a family of models used to represent networks
from the real world, although their suitability as such
is questionable for reasons that will be
elaborated below. 
Nevertheless, they are well-suited to analysis and 
as null models.

The simplest random graph is
the binomial model introduced by Erd\H{o}s and R\'enyi \citep{bollobas2001a}.
Define a probability space $\mathcal{G}(N, p)$ 
of graphs on $N$ vertices where any of the $\binom{N}{2}$ 
possible edges are chosen independently with probability $p$.
The expected average degree in the network will be $\avgdeg = p (N-1)$.
We will often make statements about this network in the 
limit of large system size $N \to \infty$. 
(In physics, this is referred to as the 
thermodynamic limit. 
Many expressions are simplified in performing this approximation,
and derived results can be quite accurate given a large enough,
albeit finite, $N$.)
The degree distribution is
a binomial distribution
$p_k = \binom{N-1}{k} p^k (1-p)^{N-1-k}$
which is well-approximated by the Poisson distribution 
$p_k \approx (\avgdeg)^k \exp({-\avgdeg})/k!$
with parameter $\avgdeg = Np$ when $N$ is large and $\avgdeg$ fixed.
For this reason, $\mathcal{G}(N,p)$ is also called the 
{\it Poisson random graph} model.
As noted previously, many networks observed in the real world
have heavy-tailed distributions, so
the Poisson model is not suitable for those types of networks.
Furthermore, real networks often contain a high density of triangles
or clustering --- friendship tends to be transitive. 
Poisson random graphs, on the other hand, 
have zero clustering in the thermodynamic limit.

A more flexible generalization
is the configuration model. 
In this model, either the degree sequence
\citep{molloy1995a, molloy1998a}
or the expected degree sequence
\citep{chung2002a} is given in advance \citep[also see][]{newman2003a};
the models produce similar networks with subtle differences.
The configuration model of 
can be thought of as a random wiring process as follows. 
First, draw a degree sequence $k_1, \ldots, k_N$
independently according to the desired degree distribution $p_k$.
We assign $k_i$ {\it stubs} (half-edges) to
each vertex $i \in V$. Next, choose a pair of stubs at random,
connect them, remove both from the queue of stubs, and continue 
until all stubs are connected. 
Finally, pairs of edges are chosen and their endpoints shuffled in
the manner of \citet{milo2002a} to ensure uniform
sampling from all graphs with the given degree sequence.
The resulting graph will have the imposed degree sequence, so long
as the sum of the degrees is even.
There may be self-loops or repeated edges, neither of which
are allowable for simple graphs. However, we expect these to occur
increasingly rarely for large $N$, and they can be removed without affecting
the resulting graph much.
Degree-degree correlations can be introduced by shuffling edges as
described in \citet{melnik2011a} and \citet{payne2011a}.
The configuration model also lacks triangles and higher-order
cliques in the thermodynamic limit, although some work has been done
to create random graph models with clustering 
\citep[see][\S\S4-5]{hackett2011a}.

\subsection{Dynamical processes on networks}

Networks provide a structure on top of which all kinds of
dynamical processes may take place.
In many cases, the structure of the network itself heavily influences
the dynamics. 
For instance, ferromagnetic materials can be modeled
as atoms in a lattice network, where the bulk magnetization
is the result of interactions between
individual atoms' magnetic moments; 
strikingly different
results occur in different dimensional lattices. 
Another example is a food web, a network of species connected
by trophic interactions (who eats whom). 
The populations of the species in the ecosystem
can be described by dynamical equations that reflect the structure of
the network, 
and this can affect ecosystem stability.
These are just two examples where dynamical processes on 
networks are a reasonable way to model system behavior.
Since networks are ubiquitous structures, these models 
appear across all disciplines \citep{vespignani2012a}.

\subsubsection{Random Boolean networks}

\label{sec:rbn}

Consider Boolean or binary state dynamics, where
each node can be either ``on'' or ``off''
(in various contexts this can mean
active/inactive, infected/susceptible, or spin up/spin down). 
The state of node $i$
is encoded by a variable $x_i \in \zeroone$, and the system
state is $\mathbf{x} = (x_i)$.
At each time step, nodes receive input from their neighbors
in the (undirected) network. 
They then compute a function of that input,
i.e.,\ 
$f_i(x_{j_1}, x_{j_2}, \ldots, x_{j_{k_i}})$ 
where
$j_1, \ldots, j_{k_i}$ 
are the neighbors of node $i$.
This determines their state in the next time step.
Because the state space $\zeroone^N$ is finite, all
trajectories are eventually periodic. 
The detailed structure of the cycles depends 
on the specific details of the network and the Boolean functions.

\newcommand{\prbn}{\tilde{p}}

The system described above, for general 
$f_i: \zeroone^{k_i} \to \zeroone$,
is known as a Boolean network. 
Boolean networks were first
studied by \citet{kauffman1969a} as a model for
dynamical behavior within cells.
See the review by \citet{aldana2003a}.
Most researchers have considered
the dynamics on random $K$-regular graphs with a 
parameter $\prbn$ that determines the bias between 0s and 1s
in the output of the update functions $f_i$, 
which are otherwise randomly chosen Boolean functions.

As mentioned, deterministic Boolean network models must be
eventually periodic. 
However, the behavior of the transient and the structure of the
basins of attraction are different for different parameters
$K$ and $\prbn$. In particular, there is a critical value
of the connectivity $K_c (\prbn)$ that separates the 
transient dynamics into two phases 
\citep{aldana2003a}:
\begin{enumerate}
\item Frozen, $K < K_c$: The distance between 
  nearby trajectories $\mathbf{x}(t)$ and $\mathbf{x}'(t)$ 
  decays exponentially with time.
\item Critical, $K = K_c$: The temporal evolution of distance between 
  trajectories is determined by fluctuations.
\item Chaotic, $K > K_c$: The distance between nearby trajectories 
  grows exponentially with time. 
\end{enumerate}


\subsubsection{Social models}

When modeling social systems with a Boolean network,
the nodes represent people and their states encode whether
or not they participate in a behavior, possess a certain belief,
etc. This could be rioting or not rioting \citep{granovetter1978a},
buying a particular style of tie \citep{granovetter1986a}, 
liking a particular band or style of music, or
believing in some unintuitive or controversial idea, 
e.g.\ climate change.
The state can represent any behavior with only two
mutually exclusive possibilities.

The function $f_i$ that determines how node $i$ changes
state is called its {\it response function} in 
sociological contexts.
\citet{schelling1971a, schelling1973a} and \citet{granovetter1978a}
pioneered the use of threshold response functions in models
of collective social behavior 
(although they were not the first; see the citations in their papers).
This was based on the intuition that, 
for a person to adopt some new behavior, 
the fraction of the population exhibiting
that behavior might need to exceed some critical value, 
the person's threshold.
(Mathematically,
a threshold response function 
$f(\phi; \onthreshold)$ 
with threshold $\onthreshold$
returns 0 if $\phi < \onthreshold$ and 1 if 
$\phi \geq \onthreshold$\footnote{The edge case 
could be defined differently, but this will
not influence the dynamics except in carefully constructed
``pathological'' scenarios.}.)
These models were generalized to the case where the dynamics take place
on a network by \citet{watts2002a}.
If we initialize a social network with some fraction of active nodes,
some of their neighbors' thresholds may be exceeded and the 
activity can spread (depending on the distribution of thresholds
and network structure).  

Standard threshold models have simple dynamical behavior:
in a word, ``spreading.'' 
If a single activation, on average, leads to more than one
subsequent activation,
then the spreading will be successful.
The activity will increase in a sigmoid fashion until some final 
fraction of the network is activated. 
These kinds of spreading are often studied using branching processes.
See the book by \citet{harris1963a} for an overview of branching 
processes and the widely-used generating function formalism. 
Branching processes have been used to model, among other things,
extinction of families, species, and genes, neutron cascades 
(as happens during nuclear chain reactions), and high energy particle
showers caused by cosmic rays.
The generating function formalism developed in part by Newman
\citep[as in][]{newman2003a, watts2002a} to analyze spreading
on networks is a straightforward application of the classical 
theory of branching processes.

We note that threshold 
random Boolean network models have been studied
for the purpose of modeling
neural networks \citep{aldana2003a}. 
Those models take place on signed, weighted
graphs, which differ from the networks considered here, and
the problems considered are different.
It would be interesting to explore the connections between our 
on-off threshold model (Section~\ref{sec:model}) 
and other threshold Boolean network models.

\subsection{Some notation}

The Bachmann-Landau asymptotic notations are used throughout this thesis.
When used carefully, asymptotic notation greatly improves the
readability of analytic statements and proofs. 
It is also widely used in probability.
For an overview of the notation's history and usage,
see \cite{knuth1976a} and
the references therein.
The notations we have used here are 
[citing from \citet{knuth1976a}]:
\begin{itemize}
\item $O(f(n))$ is the set of all $g(n)$ such that there exist positive
  constants $C$ and $n_0$ with $|g(n)| \leq C f(n)$ for all $n \geq n_0$.
\item $\Omega(f(n))$ is the set of all $g(n)$ such that there
  exist positive constants $C$ and $n_0$ with 
  $g(n) \geq C f(n)$ for all $n \geq n_0$.
\item $o(f(n))$ is the set of all $g(n)$ such that
  $g(n)/f(n) \to 0$ as $n \to \infty$.
\item $g(n) \sim f(n)$ if $g(n)/f(n) \to 1$ as $n \to \infty$.
\end{itemize}
Formally, each of the above define sets of functions, 
but we often use statements such as 
``$f$ is $O(g)$'' 
(read as ``$f$ is big-oh of $g$'') 
or ``$f = O(g)$''
to mean ``$f \in O(g)$.''

When expressing probabilities we will use notation
of the form $P(x)$, $P(x | y)$,
etc. Here, $P(x)$ is the probability that 
the random variable associated
with $x$ equals the specific value $x$. 
We leave out the random variables to avoid introducing unnecessary clutter.

Symbols in boldface represent vector quantities or 
vector-valued functions, e.g., $\mathbf{x} = (x_i)$.
Subscripts have been left out in places for clarity.
  \section{The on-off threshold model}

\label{sec:model}

Here we study a simple extension of the classical threshold models 
\citep[such as][among others]{schelling1971a, schelling1973a, 
  granovetter1978a, watts2002a,dodds2004a}: 
the response function also includes an off-threshold. 
See Figure~\ref{fig:onoffthresh} for an example on-off
threshold response function.
This is exactly the model of 
\citet{granovetter1986a}, but on a network.
We motivate this choice with the following
\citep[also see][]{granovetter1986a}.
(1) Imitation: the on state becomes favored as the fraction of active 
neighbors surpasses the on-threshold (bandwagon effect).
(2) Non-conformity: the on state is eventually less favorable 
with the fraction of active neighbors past the off-threshold
(reverse bandwagon, snob effect).
(3) Simplicity: in the absence of any raw data of ``actual''
response functions, which are surely highly context-dependent and 
variable, we choose arguably the simplest
deterministic functions which capture
imitation and non-conformity.

Let $\mathscr{G} = (V, E)$ be a graph with $N=|V|$. Assign 
each vertex $i \in V$ an on-threshold $\onthresholdi$ and an 
off-threshold 
$\offthresholdi$
with 
$0 \leq \onthresholdi \leq \offthresholdi \leq 1$.
Then that node's response function 
$f_i(\phi_i; \onthresholdi, \offthresholdi)$
is 1 if 
$\onthresholdi \leq \phi_i \leq \offthresholdi$ 
and 0 otherwise.
Let $\xvec(0) \in \zeroone^N$ be the initial states of all nodes.
At time step $t$,
each node $i$ computes the fraction 
$\phi_i(t)$ 
of their neighbors in 
$\mathscr{G}$
who are active and takes the state 
$x_i(t+1)
= f_i \left( \phi_i(t); \onthresholdi, \offthresholdi \right)$
at the next time step. 
The above defines a deterministic dynamical system
for a fixed graph and
fixed thresholds.

We now make some quick remarks about the on-off threshold model.
First, our model is a particular kind of 
Boolean network (Section~\ref{sec:rbn}).
Note that each node reacts only to the fraction of its neighbors
who are active, rather than the absolute number, and
the input varies from 0 to 1 in steps 
of $1/k_i$, where $k_i$ is node $i$'s degree.
Note that if $\onthresholdi = 0$ the node activates spontaneously, and if
$\offthresholdi = 1$ we have the usual kind of threshold response function 
(without an off-threshold).

A crucial difference between our model and many related
threshold models is that, in those models, 
an activated node can never reenter the susceptible state.
\citet{gleeson2007a} call this the permanently active property
and elaborate on its importance to their analysis.
Such models must eventually reach a steady state. 
When the dynamics are deterministic, this will
be a fixed point, and in the presence of stochasticity the steady
state is characterized by some fixed fraction of active nodes
subject to fluctuations. 
The introduction of the off-threshold builds in a mechanism for node
deactivation. 
Because nodes can now recurrently transition
between on and off states,
the deterministic dynamics can exhibit
a chaotic transient (see Section~\ref{sec:rbn}),
and the long time behavior
can be periodic with potentially high period.
With stochasticity, the dynamics can be truly chaotic and never repeat.

In the rest of this Section, we will describe how this
model is different from the random Boolean networks in the literature.
This is mainly due to the on-off threshold response functions we consider,
but also the type of random graph on which the dynamics
take place,
varying amounts of stochasticity which we introduce in the
networks and response functions, 
and the possibility of asynchronous updates.

\begin{figure}
  \centering
  \includegraphics[width=.5\linewidth]{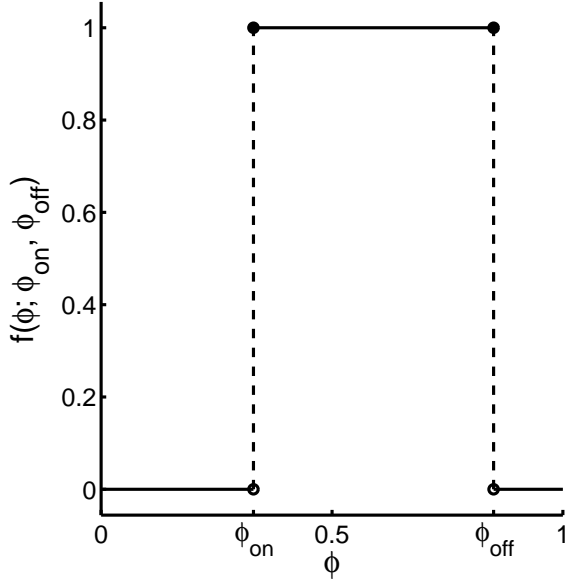}
  \caption[An example on-off threshold response function.]{
    \label{fig:onoffthresh}
    An example on-off threshold response function. 
    Here, 
    $\onthreshold = 0.33$ 
    and
    $\offthreshold = 0.85$. 
    The node will activate if 
    $\onthreshold \leq \phi \leq \offthreshold$, 
    where $\phi$
    is the fraction of its neighbors who are active. 
    Otherwise it turns off.
  }
\end{figure}

\subsection{The networks considered}

The mean field analysis in Section~\ref{sec:meanfield} 
is applicable to any network which can be characterized
by its degree distribution.
As mentioned before, the vast majority of the theory of
random Boolean networks assumes a regular random graph.
Fortunately, such theories are easily generalized to other types
of graphs with independent edges,
such as Poisson and configuration model random graphs.
Some specific results are given for
the Poisson random graph $\mathcal{G}(N, \avgdeg/N)$, and these
are the networks considered in Section~\ref{sec:tentmapf}.

\subsection{Stochastic variants}

The specific graph 
and node thresholds 
determine exactly which behaviors are possible.
These are chosen from some distribution of graphs, such as
$\mathcal{G}(N, \avgdeg/N)$, 
and some distribution of thresholds,
given by the joint density 
$P( \onthreshold , \offthreshold )$.
The specific graph and thresholds define a realization of the model
\citep[see][]{aldana2003a}.
When these are fixed for all time, 
we have, in principle, full knowledge of
the possible model dynamics. 
Given an initial condition $\xvec(0)$,
the dynamics $\xvec(t)$ 
are deterministic and known for all $t \geq 0$.

With the introduction of noise, the system is no longer
eventually periodic. Fluctuations at the node level allow
a greater exploration of state space, and the behavior is 
comparable to that of the general class of discrete-time maps.
Roughly speaking, the mean field theory we develop in 
Section~\ref{sec:meanfield} becomes more accurate as we introduce
more stochasticity.

We introduce randomness in two parts of the model:
the network and/or the response functions.
Allowing for the network and responses to be either
fixed for all time or resampled each time step and 
taking all possible combinations yields
four different designs 
(see Table~\ref{tab:modeldesigns}).

\begin{table}
  \centering
  \begin{tabular}{l|c|c}
    & {\bf R}ewiring network & {\bf F}ixed network \\
    \hline
    {\bf P}robabilistic response & P-R & P-F \\
    {\bf D}eterministic response & D-R & D-F
  \end{tabular}
  \caption[The four different ways the model can be realized.]
  {The four different ways the model can be realized. These
    are the combinations of fixed or rewired networks and
    probabilistic or deterministic response functions.
    In the thermodynamic limit of the fully stochastic version (P-R), 
    where the graph and response functions
    change every time step, the mean field theory 
    is exact (see Sec.~\ref{sec:meanfield}).}
  \label{tab:modeldesigns}
\end{table}

\subsubsection{Rewired graphs}
\label{sec:stochnet}

First, the network itself can change every time step.
This is the rewiring (R), as opposed to fixed (F),
network case.
For example, 
we could draw a new graph from $\mathcal{G}(N, \avgdeg/N)$
every time step. 
This amounts to rewiring the links while keeping
the degree distribution fixed, 
and it is alternately known as a mean 
field, ``annealed'', or random mixing variant of the 
fixed or ``quenched'' model
\citep{aldana2003a}.
\subsubsection{Probabilistic responses}
\label{sec:stochfunc}

Second, the response functions can change every time step.
This is the probabilistic (P),
as opposed to the 
deterministic (D), response function case.
Again, we will need a well-defined distribution 
$P( \onthreshold , \offthreshold )$
for the thresholds.
This amounts to having a single
response function, the expected response function
\begin{equation}
  \label{eq:fstoch}
  f (\phi) = \int d\onthreshold \int d\offthreshold \; 
  P( \onthreshold , \offthreshold ) 
  f(\phi; \onthreshold, \offthreshold) .
\end{equation}
We call $f: [0,1] \to [0,1]$ the {\it probabilistic response function}.
Its interpretation is the following.
For an updating node with a fraction $\phi$ of
active neighbors at the current time step,
then, at the next time step, the node assumes the 
state 1 with probability $f(\phi)$ and the 
state 0 with probability $1-f(\phi)$.

\subsubsection{The concept of ``temperature'' in the system}

In this thesis, the network and response functions are either
fixed for all time or resampled every time step. 
One could tune smoothly between the two extremes by
introducing rates at which these reconfigurations occur.
These rates are inversely related to quantities that
behave like temperature 
(one for the network and another for the response functions).
For the fixed case,
the temperature is zero, 
and there are no fluctuations,
while in the stochastic case, 
the temperature is very large or infinite, 
and fluctuations occur every time step.

\subsection{Synchronicity of the update}

Finally, we introduce a parameter $\alpha$ for the probability that
a given node updates. When $\alpha = 1$, all nodes update every
time step, and the update rule is said to be synchronous.
When $\alpha = 1/N$, only one node is expected to update with each
time step, and the update rule is said to be effectively asynchronous.
This is equivalent to a randomly ordered sequential update.
For intermediate values, $\alpha$ is the expected fraction of
nodes which update each time step.

\section{Fixed networks}

\label{sec:fixedanalysis}

Take the case where the response functions and
graph are fixed (D-F), but the update may be synchronous or asynchronous.
Let $x_i(t)$ 
be the probability that node $i$ is in state 1 at time $t$,
and let 
$f_i (\phi) = f_i(\phi; \onthresholdi, \offthresholdi)$.
The dynamics follow the master equation
\begin{equation}
  \label{eq:quenchedmap}
  x_i (t+1) = 
  \alpha
  f_i \left( 
    \frac{ \sum_{j=0}^N A_{ij} x_j (t)}{\sum_{j=0}^N A_{ij}}
  \right)
  + (1-\alpha) x_i(t),
\end{equation}
which can be written in matrix-vector notation as 
\begin{equation}
  \label{eq:quenchedmapmatrix}
  \xvec (t+1) = \alpha \mathbf{f} \left( T \xvec(t) \right) 
  + (1-\alpha) \xvec(t).
\end{equation}
Here $T = D^{-1}A$ is sometimes called the transition probability
matrix (in the context of a random walker), 
$D$ is the diagonal degree matrix, and 
$\mathbf{f} = (f_i)$\footnote{
  Eqns.~\eqref{eq:quenchedmap} and \eqref{eq:quenchedmapmatrix} are not
  entirely correct
  when there are isolated nodes. In that case,
  $k_i = \sum_j A_{ij} = 0$ for certain $i$,
  thus the denominator in \eqref{eq:quenchedmap} is zero
  and $D^{-1}$ undefined.
  If the initial network contains isolated nodes, 
  we set all entries in the corresponding rows of $T$ to zero.
}.
Note that if $\alpha = 1$ we recover the fully deterministic
response function dynamics, and 
$x_i(t) = 0$ or 1 for all $t$.


\subsection{Asynchronous limit}

Here, we show that when $\alpha \approx 1/N$, 
time is effectively continuous and the dynamics can be described by an 
ordinary differential equation. This is similar to the analysis of 
\citet{gleeson2008a}.
Consider Eqn.~\ref{eq:quenchedmapmatrix}. Subtracting $\xvec(t)$ from
both sides and setting $\Delta \xvec(t) = \xvec(t+1)-\xvec(t) $ and 
$\Delta t = 1$
yields
\begin{equation}
  \frac{\Delta \xvec(t)}{\Delta t} = 
  \alpha \left( \mathbf{f}( T \xvec(t) ) - \xvec(t) \right).
\end{equation}
Since $\alpha$ is assumed small, the right hand side is small, and 
thus $\Delta \xvec(t)$ is also small.
Making the continuum approximation 
$d\xvec(t)/dt \approx \Delta \xvec(t)/\Delta t$
yields the differential equation
\begin{equation}
  \label{eq:diffeq}
  \frac{d \xvec}{d t} = 
  \alpha \left( \mathbf{f}( T \xvec ) - \xvec \right).
\end{equation}
The parameter $\alpha$ sets the time scale for the system.
From their form, 
similar asynchronous, continuous time limits
apply to the dynamical equations in the densely
connected case, 
Eqn.~\eqref{eq:densemap},
and in the mean field
theory, 
Eqns.~\eqref{eq:edgemap} and~\eqref{eq:nodemap}.

\subsection{Dense network limit for Poisson random graphs}

\label{sec:densepoisson}

The following result is particular to Poisson random graphs, but
similar results are possible for other random graphs with 
dense limits.
The normalized Laplacian matrix is defined as
$\mathcal{L} \equiv I - D^{-1/2} A D^{-1/2}$,
where $I$ is the identity \citep{west2001a}.
So $T = D^{-1/2} (I- \mathcal{L}) D^{1/2}$. 
By \cite{oliveira2009a}, 
when $\avgdeg$ is $\Omega(\log{N})$
there exists a typical Laplacian matrix
$\mathcal{L}^\text{typ} 
=  I_N - \mathbf{1}_N \mathbf{1}_N^\dagger /N$
[we let $\mathbf{1}_N$ denote the length-$N$ vector of ones
and ${(\cdot)}^\dagger$ the matrix transpose]
such that the actual 
$\mathcal{L} \approx \mathcal{L}^\text{typ}$
in the induced 2-norm (spectral norm) with high probability.
In this limit, if we assume uniform degrees
$k_i = \avgdeg$ for all $i \in V$,
then
$T \approx T^\text{typ} = \mathbf{1}_N \mathbf{1}_N^\dagger / N$.
So $T$ effectively averages the node states:
$T \xvec(t) \approx T^\text{typ} \xvec(t) = 
\sum_{i=1}^N x_i (t)/N \equiv \phi(t)$.
Without a subscript, $\phi(t)$  denotes
the active fraction of the network at time $t$. 
We make the above approximation in 
Eqn.~\ref{eq:quenchedmapmatrix}
and average that equation over all nodes, finding
\begin{equation}
  \phi(t+1) = \alpha  f (\phi(t)) + (1-\alpha) \phi(t) 
  \equiv \Phi(\phi(t); \alpha, f),
  \label{eq:densemap}
\end{equation}
where we have assumed that $N$ is large and
the average of nodes' individual response functions 
$\sum_{i=1}^N f_i/N$
converges in a suitable sense to the stochastic
response function $f$, Eqn.~\eqref{eq:fstoch}.
This amounts to assuming a law of large numbers for the response functions, 
i.e., that the sample average converges to the expected function.
Note that $\alpha$ tunes between the probabilistic response
function  $\Phi(\phi;1) = f(\phi)$
and the 45$^{\circ}$ line $\Phi(\phi; 0) = \phi$.
Also, the fixed points of $\Phi$ are fixed points of
$f$, but their stability will depend on $\alpha$.

When the network is dense, it ceases to affect the dynamics, since
each node sees a large number of other nodes. 
Thus the network is effectively the complete graph.
In this way we recover
the map models of \citet{granovetter1986a}.

\section{Mean field theory}

\label{sec:meanfield}

In physics, making a ``mean field'' calculation refers to replacing
the complicated interactions among many particles
by a single interaction with some effective external field.
There are analogous techniques for understanding networks
dynamics. 
Instead of considering the $|E|$ interactions
among the $N$ nodes, network mean field theories derive
self-consistent expressions for the overall behavior of the network,
after averaging over large sets of nodes.
These have been fruitful in the study of random Boolean
networks \citep{derrida1986a} and can be surprisingly 
effective when networks are non-random \citep{melnik2011a}.

We derive a mean field theory, in the thermodynamic limit,
for the dynamics of the on-off threshold model
by blocking nodes according to their degree class. 
This is equivalent to nodes retaining their degree but rewiring 
edges every time step.
The model is then part of the well-known class of random mixing
models with non-uniform contact rates.
Probabilistic (P-R) and deterministic (D-R) response functions
result in equivalent behavior for these random mixing models.
The important state variables end up being the active density of stubs.
In an undirected network without degree-degree correlations, 
the state is described by a single variable $\rho(t)$. 
In the presence of correlations we must
introduce more variables, i.e.,\ $\rho_k(t), \rho_{k'}(t)$, \ldots, 
to deal with the relevant degree classes.

\subsection{Undirected networks}

\label{sec:mfundir}

To derive the mean field equations in the simplest case ---
undirected, uncorrelated random graphs --- consider a degree $k$
node at time $t$. 
The probability that the node is in the 1 state at time $t+1$
given a density $\rho$ of active stubs is 
\begin{equation}
  \label{eq:Gmap}
  F_k ( \rho; f ) = \sum_{j=0}^{k} \binom{k}{j} \rho^j (1-\rho)^{k-j} 
  f(j/k),
\end{equation}
where each term in the sum counts the contributions 
from having 0, 1, \ldots, $k$ active neighbors.
Now, the probability of 
choosing a random stub which ends at a degree $k$ node is 
$q_k = k p_k/\avgdeg$
in an uncorrelated random network \citep{newman2003a}.
This is sometimes called 
the edge-degree distribution.
So if all of the nodes update synchronously, the active density
of stubs at $t+1$ will be
\begin{equation}
  \label{eq:gmap}
  g(\rho ; p_k, f)
  = \sum_{k=1}^{\infty} q_k F_k( \rho; f)
  = \sum_{k=1}^{\infty} \frac{k p_k}{\avgdeg} F_k( \rho; f).
\end{equation}
Finally, if each node only updates with probability $\alpha$,
we have the following map for the density of active stubs:
\begin{equation}
  \rho(t+1) = \alpha \, g \left( \rho(t); p_k, f \right)  + (1-\alpha) \rho(t)
  \equiv G(\rho(t); p_k, f, \alpha) .
  \label{eq:edgemap}
\end{equation}
By a similar argument, the active density of nodes is given by
\begin{equation}
  \label{eq:nodemap}
  \phi(t+1) = \alpha \, h(\rho(t); p_k, f ) + (1-\alpha) \phi(t)
  \equiv H( \rho(t), \phi(t); p_k, f, \alpha),
\end{equation}
where
\begin{equation}
  h( \rho; p_k, f) = \sum_{k=0}^{\infty} p_k F_k( \rho; f) .
\end{equation}
Note that the edge-oriented state variable $\rho$
contains all of the dynamically important information,
rather than the vertex-oriented variable $\phi$.

\subsection{Analysis of the map equation}
\label{sec:mfanalysis}

The function $F_k(\rho; f)$ is known in polynomial approximation theory
as the $k$th Bernstein polynomial (in the variable $\rho$) of $f$
\citep{phillips2003a, pena1999a}. 
These are approximating polynomials which have applications
in computer graphics due to their ``shape-preserving properties.''
The Bernstein operator $\mathbb{B}_k$ takes 
$f \mapsto F_k$. 
This is a linear, positive operator which preserves convexity
for all $k$ and exactly interpolates the endpoints $f(0)$ and $f(1)$. 
Immediate consequences include that
each $F_k$ is a smooth function and the $k$th derivatives
$F_k^{(k)}(x) \to f^{(k)}(x)$ where $f^{(k)}(x)$ exists.
For concave $f$ 
(such as the tent or logistic maps), 
we have concave $F_k$ for all $k$ and $F_k \nearrow f$ uniformly.
This convergence is typically slow. 
Importantly, $F_k \nearrow f$ implies that 
$g(\rho; p_k, f) \leq f$  
for any degree distribution $p_k$.

In some cases,
the dynamics of the undirected mean field theory
given by $\rho(t+1)=G(\rho(t))$, 
Eqn.~\eqref{eq:edgemap},
are effectively those of the map $\Phi$, 
from the dense limit Eqn.~\eqref{eq:densemap}.
We see that  $g$, Eqn.~\eqref{eq:gmap},
can be seen as the expectation of a sequence
of random functions $F_k$
under the edge-degree distribution $q_k$
(indeed, this is how it was derived).
From the convergence of the $F_k$'s,
we expect that
$g(\rho; p_k, f) \approx f(\rho)$ 
if the average degree $\avgdeg$ 
is ``large enough'' and 
the edge-degree distribution has a ``sharp enough'' peak about $\avgdeg$
(we will clarify this soon). 
Then as $\avgdeg \to \infty$,
the mean field coincides with the
dense network limit we found for Poisson random graphs, 
Eqn.~\ref{eq:densemap}.
Some thought leads to a sufficient condition for
this kind of convergence: the standard deviation $\sigma(\avgdeg)$
of the degree distribution must be $o(\avgdeg)$.
In Appendix~\ref{ax:limitthm} we prove this as Lemma~\ref{thm:1}.

In general, if the original degree distribution $p_k$ is
characterized by having mean $\avgdeg$, variance $\sigma^2$, and 
skewness $\gamma_1$, then the edge-degree distribution $q_k$
will have mean $\avgdeg + \sigma^2/\avgdeg$ and variance
$\sigma^2 [ 1 + \gamma_1 \sigma/\avgdeg - (\sigma/\avgdeg)^2]$.
Considering the behavior as $\avgdeg \to \infty$, we can conclude that
requiring $\sigma \in o(\avgdeg)$ and $\gamma_1 \in o(1)$ are sufficient
conditions on $p_k$ to apply Lemma~\ref{thm:1}.
Poisson degree distributions ($\sigma = \sqrt{\avgdeg}$ and
$\gamma_1 = \avgdeg^{-1/2}$) fit
these criteria.

\subsection{Generalized random networks}
\label{sec:mixednets}

In the most general kind of random networks, 
edges can be undirected or directed,
and we then denote
node degree by a vector
$\kvec = (\kbid, \kin, \kout)^\dagger$. 
The degree distribution is written as
$p_\kvec \equiv P(\kvec)$.
There may also
be correlations between node degrees.
Correlations of this type are encoded by the conditional 
probabilities
\begin{gather*}
  p^{(\bidmark)}_{\kvec,\kvec'} \equiv
  P( \kvec, \mathrm{undirected} | \kvec') \\
  p^{(\inmark)}_{\kvec,\kvec'} \equiv
  P( \kvec, \mathrm{incoming} | \kvec') \\
  p^{(\outmark)}_{\kvec,\kvec'} \equiv
  P( \kvec, \mathrm{outgoing} | \kvec'),
\end{gather*}
the probability that an edge starting at a degree $\kvec'$
node ends at a degree $\kvec$ node and is, respectively,
undirected, incoming, or outgoing relative to the destination degree
$\kvec$ node. 
We introduced this convention in a series of papers 
\citep{payne2011a, dodds2011a}.
These conditional probabilities can also be defined
in terms of the joint distributions of node types connected by
undirected and directed edges.
The mean field equations for the on-off threshold model are
closely related to the equations for the time evolution of
a contagion process \citep[][Eqns.~(13--15)]{payne2011a}.
We omit a detailed derivation, since it
is similar to that in Section~\ref{sec:mfundir}
\citep[see also][]{gleeson2007a, payne2011a}.
The result is a coupled system of equations for the
density of active stubs which now may depend on
node type ($\kvec$) and
edge type (undirected or directed):
\begin{align}
  \rho_{\kvec}^{(\bidmark)}(t+1) = &
  \alpha 
  \sum_{\kvec'} p^{(\bidmark)}_{\kvec,\kvec'} 
  \sum_{j_u=0}^{\kbid'} \sum_{j_i=0}^{\kin'}
  \binom{\kbid'}{j_u} \binom{\kin'}{j_i}  \nonumber \\
  & \times
  \left[ \rho_{\kvec'}^{(\bidmark)}(t) \right]^{j_u}
  \left[1 - \rho_{\kvec'}^{(\bidmark)}(t) \right]^{(\kbid'-j_u)} \nonumber \\
  & \times
  \left[ \rho_{\kvec'}^{(\inmark)}(t) \right]^{j_i}
  \left[1 - \rho_{\kvec'}^{(\inmark)}(t) \right]^{(\kin'-j_i)} 
  \: f \left( \frac{j_u + j_i}{\kbid' + \kin'} \right) \nonumber \\
  &
  + (1 - \alpha) \rho_{\kvec}^{(\bidmark)}(t) \label{eq:mfgenbid} \\
  \rho_{\kvec}^{(\inmark)}(t+1) = &
  \alpha \sum_{\kvec'} p^{(\inmark)}_{\kvec,\kvec'} 
  \sum_{j_u=0}^{\kbid'} \sum_{j_i=0}^{\kin'}
  \binom{\kbid'}{j_u} \binom{\kin'}{j_i}  \nonumber \\
  & \times
  \left[ \rho_{\kvec'}^{(\bidmark)}(t) \right]^{j_u}
  \left[1 - \rho_{\kvec'}^{(\bidmark)}(t) \right]^{(\kbid'-j_u)} \nonumber \\
  & \times
  \left[ \rho_{\kvec'}^{(\inmark)}(t) \right]^{j_i}
  \left[1 - \rho_{\kvec'}^{(\inmark)}(t) \right]^{(\kin'-j_i)} 
  \: f \left( \frac{j_u + j_i}{\kbid' + \kin'} \right) \nonumber \\
  &
  + (1 - \alpha) \rho_{\kvec}^{(\inmark)}(t) \, . \label{eq:mfgendir}
\end{align}
The active fraction of nodes at a given time is given by:
\begin{align}
  \phi(t+1) = &
  \alpha \sum_{\kvec} p_{\kvec} 
  \sum_{j_u=0}^{\kbid} \sum_{j_i=0}^{\kin}
  \binom{\kbid}{j_u} \binom{\kin}{j_i}  \nonumber \\
  & \times
  \left[ \rho_{\kvec}^{(\bidmark)}(t) \right]^{j_u}
  \left[1 - \rho_{\kvec}^{(\bidmark)}(t) \right]^{(\kbid-j_u)} \nonumber \\
  & \times
  \left[ \rho_{\kvec}^{(\inmark)}(t) \right]^{j_i}
  \left[1 - \rho_{\kvec}^{(\inmark)}(t) \right]^{(\kin-j_i)} 
  \: f \left( \frac{j_u + j_i}{\kbid + \kin} \right) \nonumber \\
  &
  + (1 - \alpha) \phi(t) \, . \label{eq:mfgennodes}
\end{align}

\section{Poisson random graphs with tent map average response function}

\label{sec:tentmapf}

The results so far have been entirely general, in the sense
that the underlying network and thresholds are arbitrary.
Now we apply the general theory to the case of Poisson 
random graphs with a simple distribution of thresholds.

The networks we consider are 
Poisson random graphs from $\mathcal{G}(N, \avgdeg/N)$.
The thresholds $\onthreshold$ and $\offthreshold$ 
are now distributed uniformly
on $[0, 1/2)$ and $[1/2, 1)$, respectively.
This distribution results in the probabilistic response function
(see Figure~\ref{fig:tentmap}):
\begin{equation}
  \label{eq:tentmap}
  f(\phi) = \left\{
    \begin{array}{lrl}
      2\phi \hfill& \mathrm{if} & 0 \leq \phi < 1/2, \\
      2-2\phi \hfill & \mathrm{if} & 1/2 \leq \phi \leq 1 .
    \end{array}
  \right.
\end{equation}
The tent map is a well-known
chaotic map of the unit interval \citep{alligood1996a}.
We thus expect the on-off threshold model with this probabilistic
response function to exhibit similarly interesting
behavior.

\begin{figure}
  \centering
  \includegraphics[width=.5\linewidth]{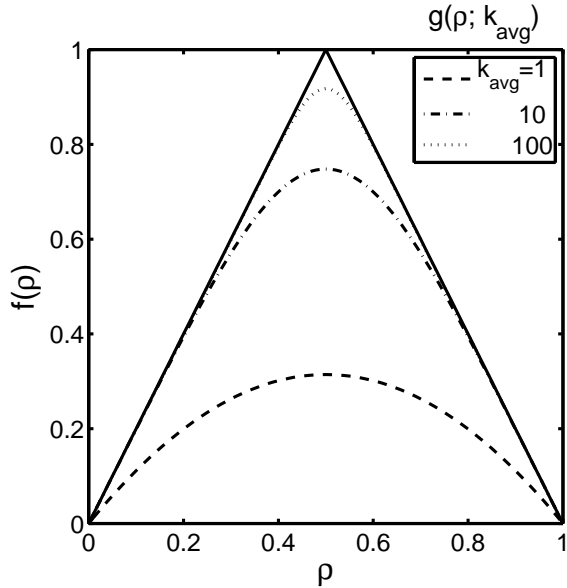}
  \caption[The tent map probabilistic response function $f(\rho)$,
  Eqn.~\eqref{eq:tentmap}.]{
    \label{fig:tentmap}
    The tent map probabilistic response function $f(\rho)$,
    Eqn.~\eqref{eq:tentmap}. 
    Note that we use the argument 
    $\rho$
    for comparison with the 
    edge maps 
    $g(\rho;\avgdeg) = g(\rho; p_k, f)$, 
    Eqn.~\eqref{eq:gmap}, shown for $\avgdeg = 1, 10, 100$.
    These $p_k$ are Poisson distributions
    with mean $\avgdeg$.
    As $\avgdeg$ increases, $g(\rho;\avgdeg)$ increases to $f(\rho)$.}
\end{figure}

\subsection{Analysis of the dense limit}
\label{sec:dense}

When the network is in the dense limit (Section~\ref{sec:densepoisson}), 
the dynamics follow
$\phi(t+1) = \Phi( \phi(t); \alpha )$, where
\begin{equation}
  \label{eq:densetent}
  \Phi(\phi; \alpha) = \alpha f(\phi) + (1-\alpha) \phi =
  \left\{
    \begin{array}{lrl}
      (1+\alpha) \phi & \mathrm{if} & 0 \leq \phi < 1/2, \\
      (1- 3\alpha)\phi + 2\alpha &  \mathrm{if} & 1/2 \leq \phi \leq 1 .
    \end{array}
  \right.
\end{equation}
Solving for the fixed points of $\Phi(\phi; \alpha)$, 
we find one at $\phi=0$ and
another at $\phi = 2/3$. When $\alpha < 2/3$, the 
nonzero fixed point is attracting for all initial conditions
except $\phi = 0$.
When $\alpha = 2/3$, $[1/2, 5/6]$ is an 
interval of period 2 centers. 
Any orbit will eventually land on one of these period 2 orbits.
When $\alpha > 2/3$, this interval
of period 2 orbits ceases to exist,
and more complicated behavior ensues. 
Figure~\ref{fig:bifurcdense} shows the bifurcation diagram 
for $\Phi(\phi; \alpha)$.
From the bifurcation diagram,
the orbit appears to cover dense subsets of the unit interval
when $\alpha > 2/3$. 
The bifurcation diagram appears like
that of the tent map 
(not shown; see \citealp{alligood1996a, dodds2012a}),
except that the branches to the right of the first bifurcation point
are separated here by the interval of period 2 orbits.

\begin{figure}
  \centering
  \includegraphics[width=0.8\textwidth]{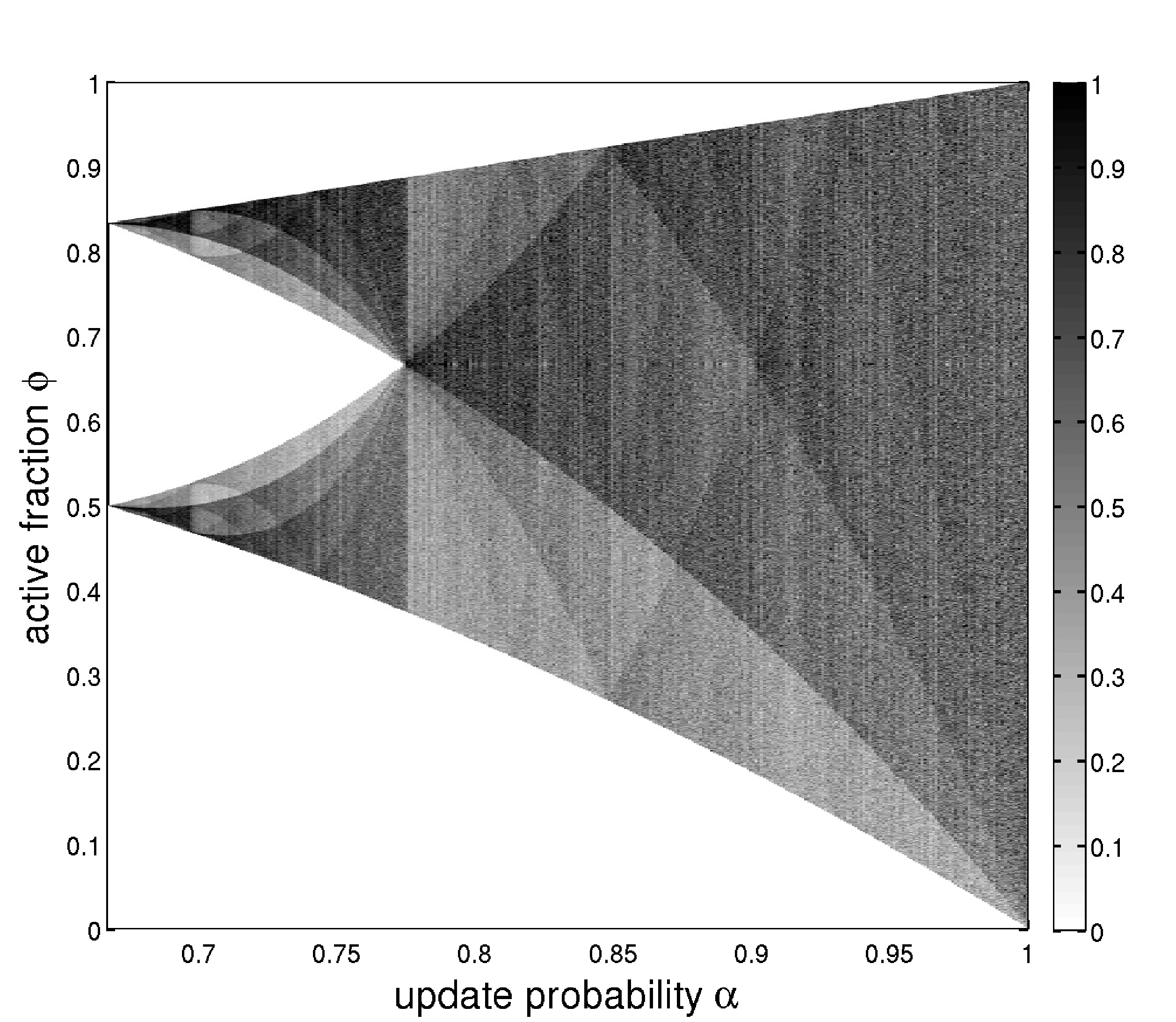}
  \caption[Bifurcation diagram for the dense map $\Phi(\phi;\alpha)$, 
  Eqn.~\eqref{eq:densetent}.]
  {Bifurcation diagram for the dense map $\Phi(\phi; \alpha)$,
    Eqn.~\eqref{eq:densetent}.
    This was generated by iterating the map at 1000
    $\alpha$ values between 0 and 1. 
    The iteration was carried out with
    3 random initial conditions for
    10000 time steps each, discarding the first 1000.
    The $\phi$-axis contains 1000 bins and the invariant density,
    shown by the grayscale value, is normalized
    by the maximum for each $\alpha$. With $\alpha < 2/3$ (not shown),
    all trajectories go to the fixed point at $\phi = 2/3$.
  }
  \label{fig:bifurcdense}
\end{figure}

\subsubsection{The effect of conformists, an aside}

Suppose some fraction $c$ of the population is made up of 
individuals without any off-threshold 
(alternatively, each of their off-thresholds $\offthreshold=1$). 
These individuals are conformist or ``purely pro-social''
in the sense that they are perfectly 
happy being part of the majority. 
For simplicity, assume $\alpha = 1$.
The map $\Phi(\phi; c) = 2 \phi$ for $0 \leq \phi < 1/2$
and $2 - 2(1-c)\phi$ for $1/2 \leq \phi \leq 1$.
If $c > 1/2$, then the equilibrium at 2/3 is stable. Pure conformists, then,
can have a stabilizing effect on the process. 
We expect a similar effect when the network is not dense.

\subsection{Mean field}

\subsubsection{Analysis}

In this specific example,
we can write the degree-dependent map 
$F_k(\rho ; f)$ 
in terms of incomplete regularized beta functions \citep{dlmf2012}. 
Since $f$ is understood to be the tent map, 
we will write $F_k(\rho; f) = F_k(\rho)$.
First, use the piecewise form of 
Eqn.~\eqref{eq:tentmap} to write
\begin{align*}
  F_k(\rho) &= 
  \sum_{j=0}^M \binom{k}{j} \rho^j (1-\rho)^{k-j} \left( \frac{2j}{k} \right) 
  + \sum_{j=M+1}^{k} \binom{k}{j} \rho^j (1-\rho)^{k-j} 
  \left( 2-\frac{2j}{k} \right) \\
  &= 2 - 2 \rho 
  - 2 \sum_{j=0}^M \binom{k}{j} \rho^j (1-\rho)^{k-j} 
  + \left( \frac{4}{k} \right)
  \sum_{j=0}^M \binom{k}{j} \rho^j (1-\rho)^{k-j} j .
\end{align*}
We have let $M=\lfloor k/2 \rfloor$ for clarity 
($\lfloor \cdot \rfloor$ and
$\lceil \cdot \rceil$ are the floor and
ceiling functions) and
used the fact that the binomial distribution
$\binom{k}{j} \rho^j(1-\rho)^{k-j}$
sums to one and has mean $k \rho$.
For $n \leq M$, we have the identity
\begin{align}
  \label{eq:factorialmoments}
  \sum_{j=0}^{M} (j)_n \binom{k}{j} \rho^j (1-\rho)^{k-j} = 
  \rho^n (k)_n I_{1-\rho}( k-M, M-n+1)
\end{align}
where $I_{x}(a,b)$ is the regularized incomplete beta function and
$(k)_n = k (k-1) \cdots (k - (n-1))$ is the falling factorial
\citep{winkler1972a, dlmf2012}. 
This is an expression for the partial (up to $M$) 
$n$th factorial moment of the binomial distribution with parameters
$k$ and $\rho$. Note that when $n=0$ we recover the well-known expression
for the binomial cumulative distribution function. 
So, 
\begin{subequations}
  \label{eq:Gmapbetafn}
  \begin{align}
    F_k (\rho) = & \, 2 \rho \, - 
    4 \rho I_{\rho} (M, k-M) + 2 I_{\rho} (M+1, k-M)
    \label{eq:Gmapbetafn1} \\
    = & \, (2 - 2 \rho) \, - 
    ( 2 I_{1-\rho} (k-M, M+1) - 4 \rho I_{1-\rho}(k-M, M) ) . 
    \label{eq:Gmapbetafn2}
  \end{align}
\end{subequations}

When $0 \leq \rho \leq 1/2$, 
the form of Eqn.~\eqref{eq:Gmapbetafn1}
can be used to show directly that 
$F_k(\rho)$ is bounded above by the tent map $f(\rho)$,
which we already knew from the properties of the
Bernstein polynomials.
A similar approach works for the region $1/2 \leq \rho \leq 1$
using Eqn.~\eqref{eq:Gmapbetafn2}.
We find a weak bound for the rate at which $F_k (\rho)$ 
converges to $f(\rho)$.
For $\rho < 1/2$, using Eqn.~\eqref{eq:Gmapbetafn1},
\begin{align*}
  f(\rho) - F_k(\rho) &= 4 \rho I_{\rho} (M, k-M) - 2 I_{\rho} (M+1, k-M) \\
  & \leq 2 \left( I_{\rho} (M, k-M) - I_{\rho} (M+1, k-M) \right) \\
  & \leq 2 \frac{ \rho^M (1-\rho)^{k-M} }{M B(M, k-M)} \\
  & \leq \frac{ 2 (1/2)^k }{ M B(M, k-M) } \\
  & \leq \frac{ 4 \Gamma(k)}{k 2^k [\Gamma(k/2)]^2} 
  \qquad \text{(for even $k$)}\\
  & \leq 2 \pi^{-1/2} \frac{1}{k} \cdot \frac{\Gamma(k/2 + 1/2)}{ \Gamma(k/2)}
  = O( k^{-1/2} ),
\end{align*}
where we have used identities from \cite{dlmf2012}. 
We find the same $O(k^{-1/2})$ behavior
using Eqn.~\eqref{eq:Gmapbetafn2} in the region $\rho > 1/2$.

Finally, note that the active edge fraction 
$\rho(t) \approx \phi(t)$,
the active node fraction.
This is because Poisson random graphs are highly regular, with
$q_k = k p_k / \avgdeg = p_{k-1} \approx p_k$.
Thus the mean field dynamics for active edge density 
are effectively the same as for active node density.

\subsubsection{Numerical algorithm}

The map $g(\rho; p_k, f)$ 
is parametrized here by the network parameter $\avgdeg$,
since $p_k$ is fixed as a Poisson distribution with mean $\avgdeg$ 
and $f$ is the tent map, and
we write it as simply $g(\rho; \avgdeg)$. 
To evaluate $g(\rho; \avgdeg)$,
we compute $F_k(\rho)$ using Eqn.~\eqref{eq:Gmapbetafn}
and constrain the sum in 
Eqn.~\eqref{eq:gmap} to values of $k$ with
$\lfloor \avgdeg-3 \sqrt{\avgdeg} \rfloor \leq k \leq \lceil \avgdeg+3 \sqrt{\avgdeg} \rceil$.
This computes contributions to 
within three standard deviations of the average degree in
the graph, requiring only $O (\sqrt{\avgdeg})$ evaluations of 
Eqn.~\eqref{eq:Gmapbetafn}. 
The representation in  
Eqn.~\eqref{eq:Gmapbetafn} allows for 
quick numerical evaluation of $F_k( \rho )$ for any $k$,
which we performed in MATLAB
using the built-in routines for the incomplete beta function.

In Figure~\ref{fig:tentmap},  
we show
$g(\rho; \avgdeg)$ 
for 
$\avgdeg = 1, 10,$ and 100.
We confirm the conclusions of Section~\ref{sec:mfanalysis}:
$g(\rho; \avgdeg)$ is bounded above by $f(\rho)$, 
and 
$g(\rho; \avgdeg) \nearrow f(\rho)$
as $\avgdeg \to \infty$. Convergence is slowest at $\rho = 1/2$, and the
kink that the tent map has there has been smoothed out by the
effect of the Bernstein operator.

\subsection{Simulations}

We performed direct simulations of the on-off threshold model
for the D-F, P-F, and P-R designs, 
in the abbreviations of Table~\ref{tab:modeldesigns}.
Unless otherwise noted, $N = 10^4$.
For all of the bifurcation diagrams,
the first 3000 time steps were considered transient and discarded,
and the invariant density of $\rho$ 
was calculated from the following 1000 points.
For plotting purposes, the invariant density was normalized by its
maximum at those parameters. 
For example, in Figure~\ref{fig:bifurcdense} we plot
$P(\phi | \alpha) / \max_\phi P( \phi | \alpha)$ 
rather than the raw density 
$P(\phi | \alpha)$.

To compare the mean field theory to those simulations,
we numerically iterated the edge map
$\rho(t+1) = G(\rho(t); \avgdeg, \alpha)$
for different values of $\alpha$ and $\avgdeg$.
We then created bifurcation
diagrams of the possible behavior in the mean field
as was done for the simulations.

\subsection{Results}

\label{sec:results}

\begin{figure}
  \centering
  \includegraphics[width=\textwidth]{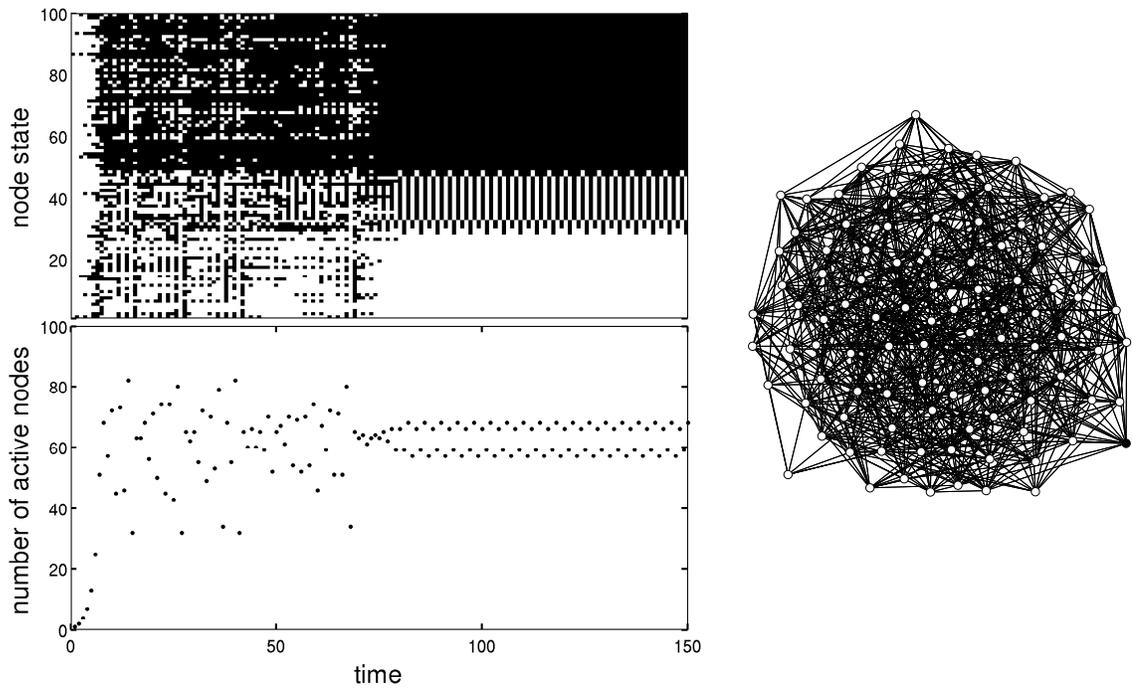}
  \caption[Deterministic (D-F) dynamics on a small graph.]
  {Deterministic (D-F) dynamics on
    a small graph. 
    Here, $N = 100$ and $\avgdeg = 17$. 
    On the left, 
    we plot the state evolution over time. 
    The upper plot shows
    individual node states (black = active) sorted by their
    eventual level of activity, 
    and the lower plot shows the total number 
    of active nodes. 
    We see that the contagion takes off, 
    followed by a transient period of unstable behavior
    until around time step 80, 
    when the system enters a macroperiod 4 orbit. 
    Note that individual nodes exhibit 
    different microperiods (see Sec.~\ref{sec:results}).
    On the right, 
    we show the network itself with the initial seed node
    in black in the lower right.}
  \label{fig:smalldet}
\end{figure}

We show results for the 
deterministic model on a small network in Figure~\ref{fig:smalldet}.
Here, $N =100$ and $\avgdeg=17$.
Starting from an initial active node at $t = 0$, the active population
grows monotonically over the next 6 time steps. 
From $t=6$ to $t \approx 80$,
the transient time, the active population fluctuates 
in a similar manner to the stochastic case. 
After $t \approx 80$, the state collapses into a period 4 orbit. 
We call the overall period of the system its ``macroperiod.''
Individual nodes may exhibit different ``microperiods.''
Note that the macroperiod is the lowest common multiple of the individual
nodes' microperiods.
In Figure~\ref{fig:smalldet}, we observe microperiods 1, 2, and 4
in the timeseries of individual node activity.
A majority of the nodes end up frozen 
in the on or off state,
with approximately $20 \%$ of the nodes exhibiting cyclical behavior
after collapse.
The focus of this thesis has been the analysis of the on-off threshold
model, and the D-F
case has not been as amenable to analysis as the 
stochastic cases.
A deeper examination through simulation
of the deterministic case will appear in \cite{dodds2012a}.

\begin{figure}
  \centering
  \includegraphics[width=0.7\textwidth]{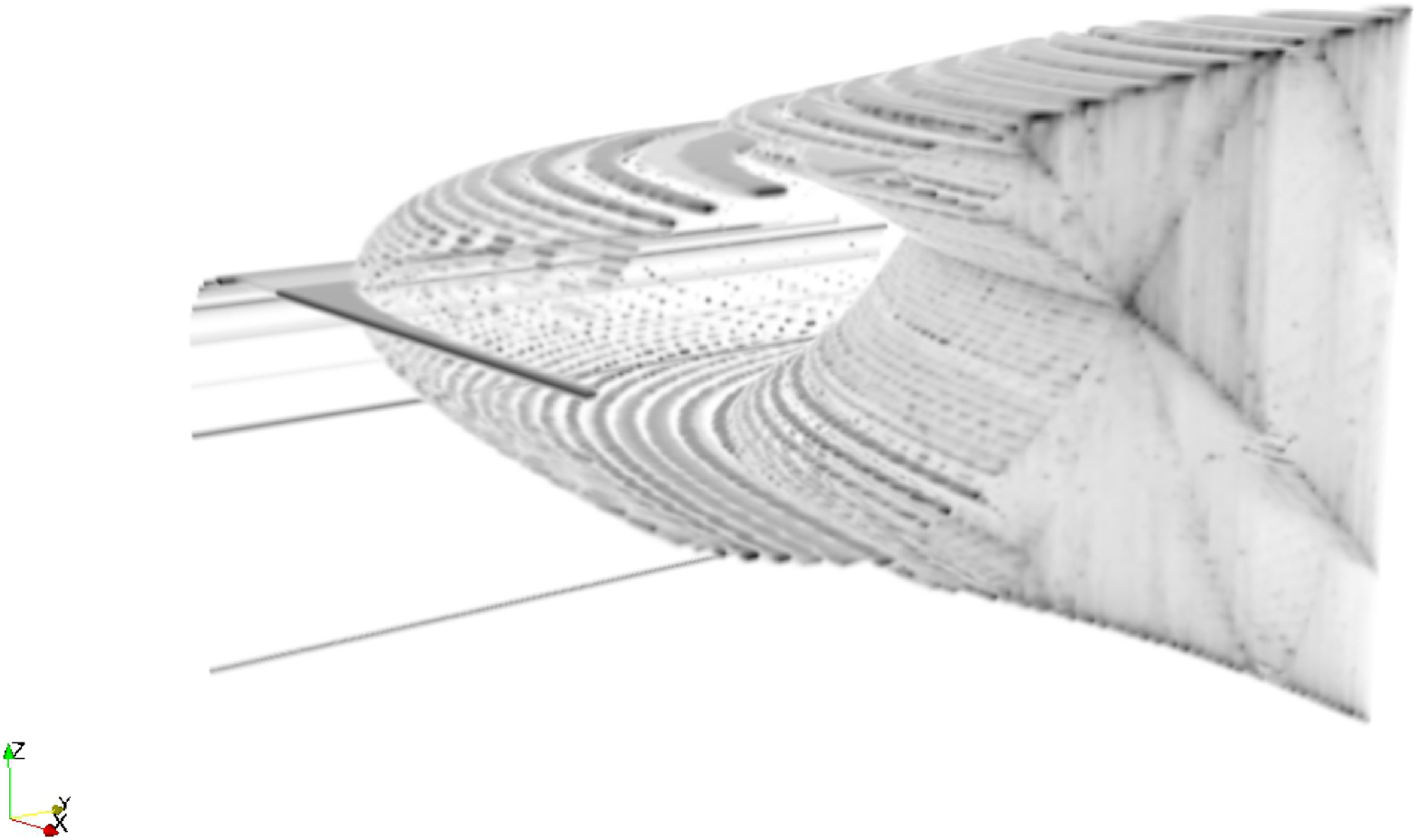}\\
  \includegraphics[width=0.7\textwidth]{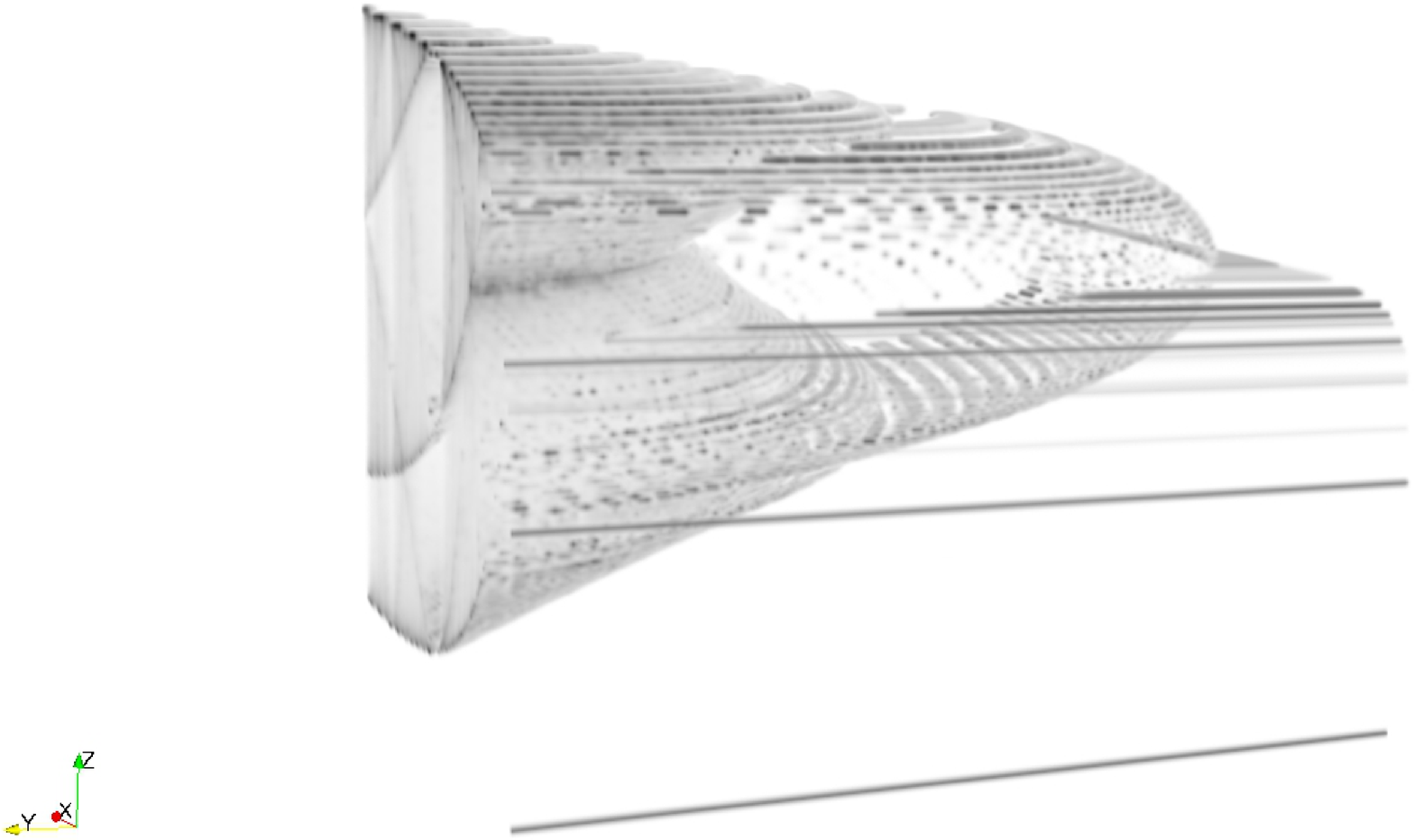}\\
  \includegraphics[width=0.7\textwidth]{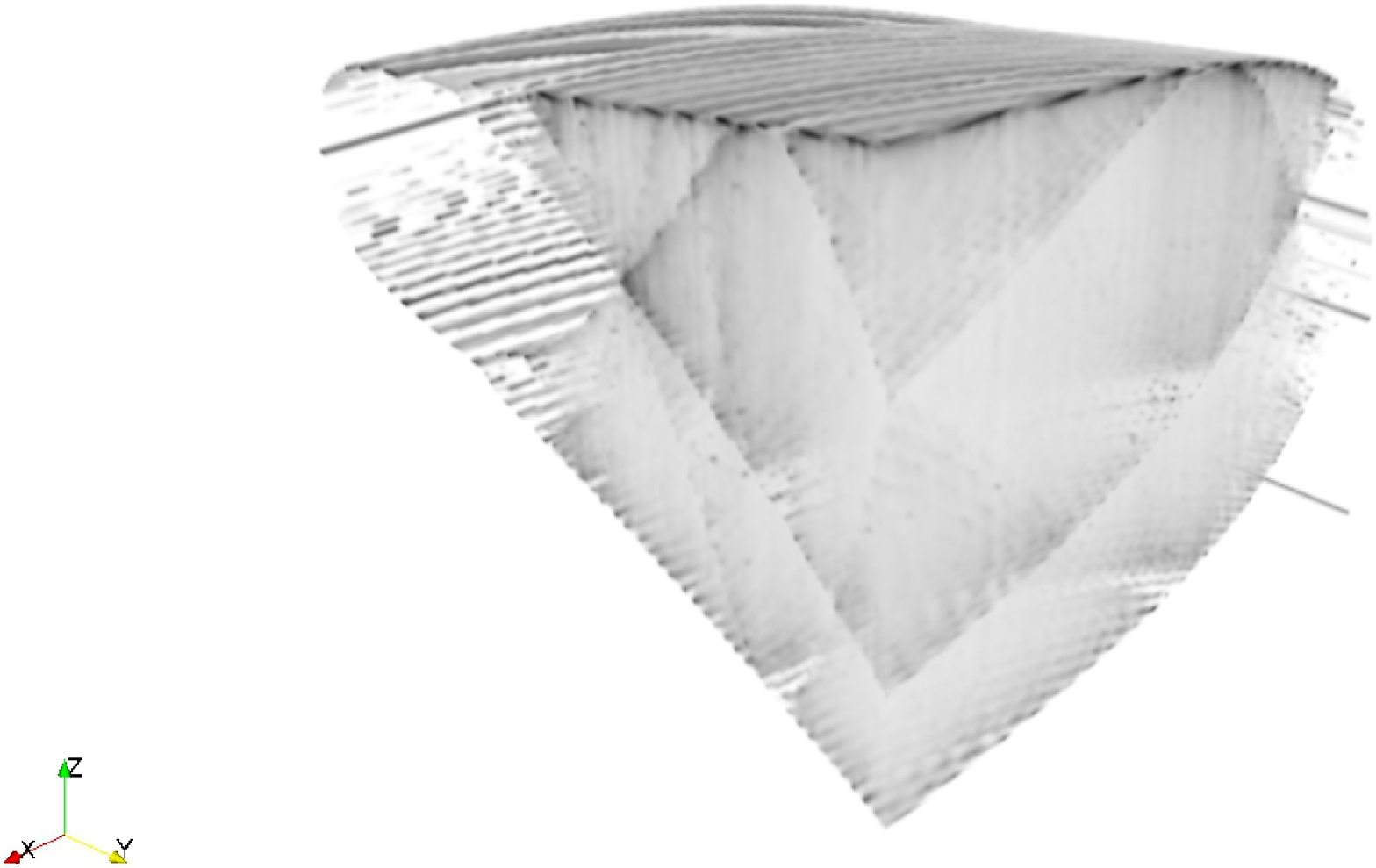}\\
  \caption[The 3-dimensional bifurcation diagram computed 
  from the mean field theory.]
  {The 3-dimensional bifurcation diagram computed 
    from the mean field theory.
    The axes 
    X = average degree $\avgdeg$, 
    Y = update probability $\alpha$, and 
    Z = active edge fraction $\rho$. 
    The discontinuities of the surface
    are due to the limited resolution of our simulations.
    See Figure~\ref{fig:bifurcslice_eqn} for the parameters used.
    This was plotted in Paraview.}
  \label{fig:3dbifurc}
\end{figure}

\begin{figure}
  \centering
  \includegraphics[width=\textwidth]
  {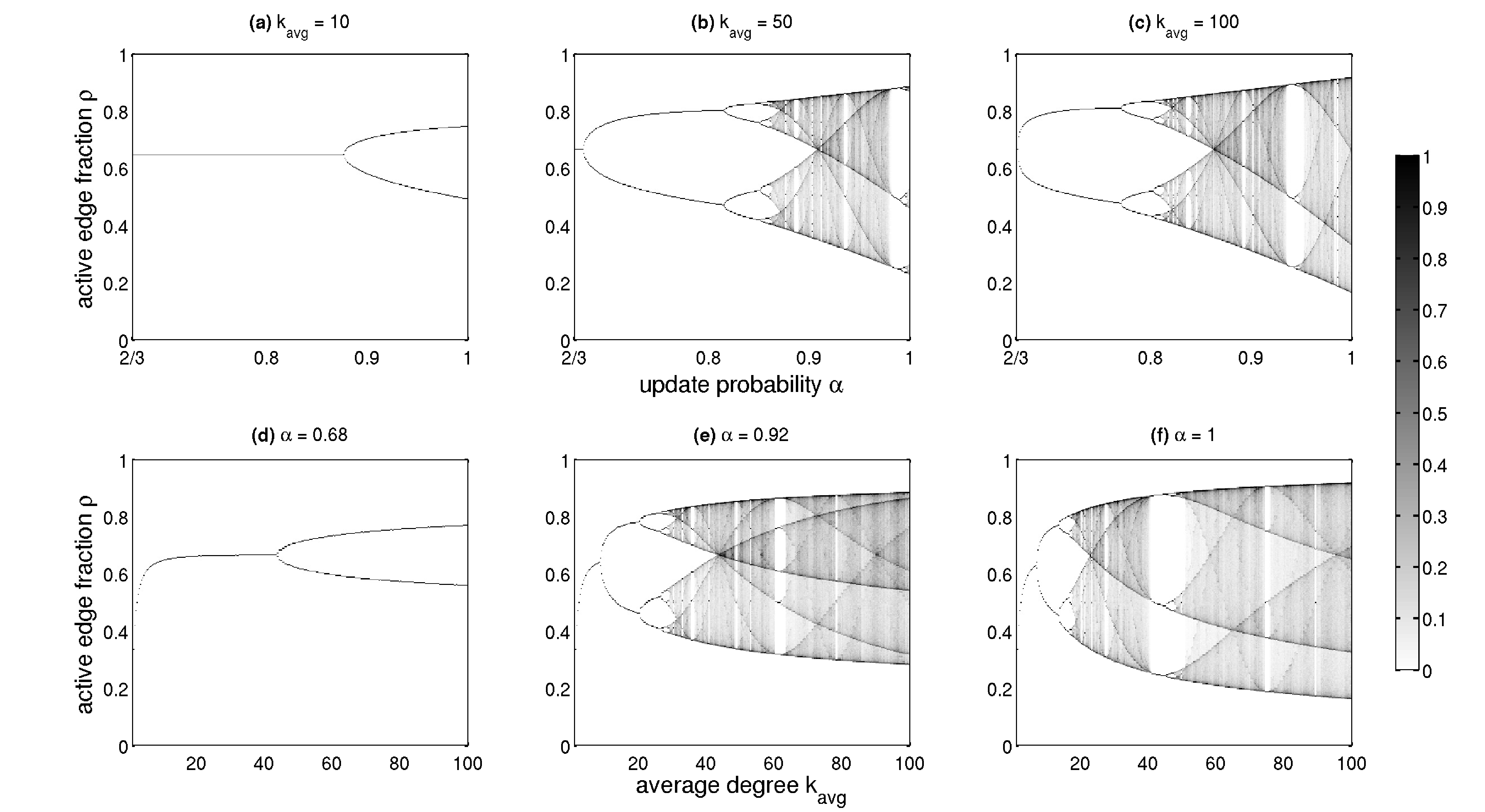}
  \caption[Mean field theory bifurcation diagram slices 
  for various fixed values of $\avgdeg$ and $\alpha$.]
  {Mean field theory bifurcation diagram slices 
    for various fixed values of $\avgdeg$ and $\alpha$.
    The top row (a--c) shows slices for fixed $\avgdeg$.
    As $\avgdeg \to \infty$, the $\avgdeg$-slice bifurcation diagram 
    asymptotically approaches the bifurcation diagram for the 
    dense map, Figure~\ref{fig:bifurcdense}. Note that the first
    bifurcation point, near 2/3, grows steeper with increasing $\avgdeg$.
    The bottom row (d--f) shows slices for fixed $\alpha$.
    The resolution of the simulations was
    $\alpha = 0.664, 0.665, \ldots, 1$, $\avgdeg=1, 1.33, \ldots, 100$,
    and $\rho$ bins were made for 1000 points between 0 and 1.
  }
  \label{fig:bifurcslice_eqn}
\end{figure}

\begin{figure}
  \centering
  \includegraphics[width=\textwidth]
  {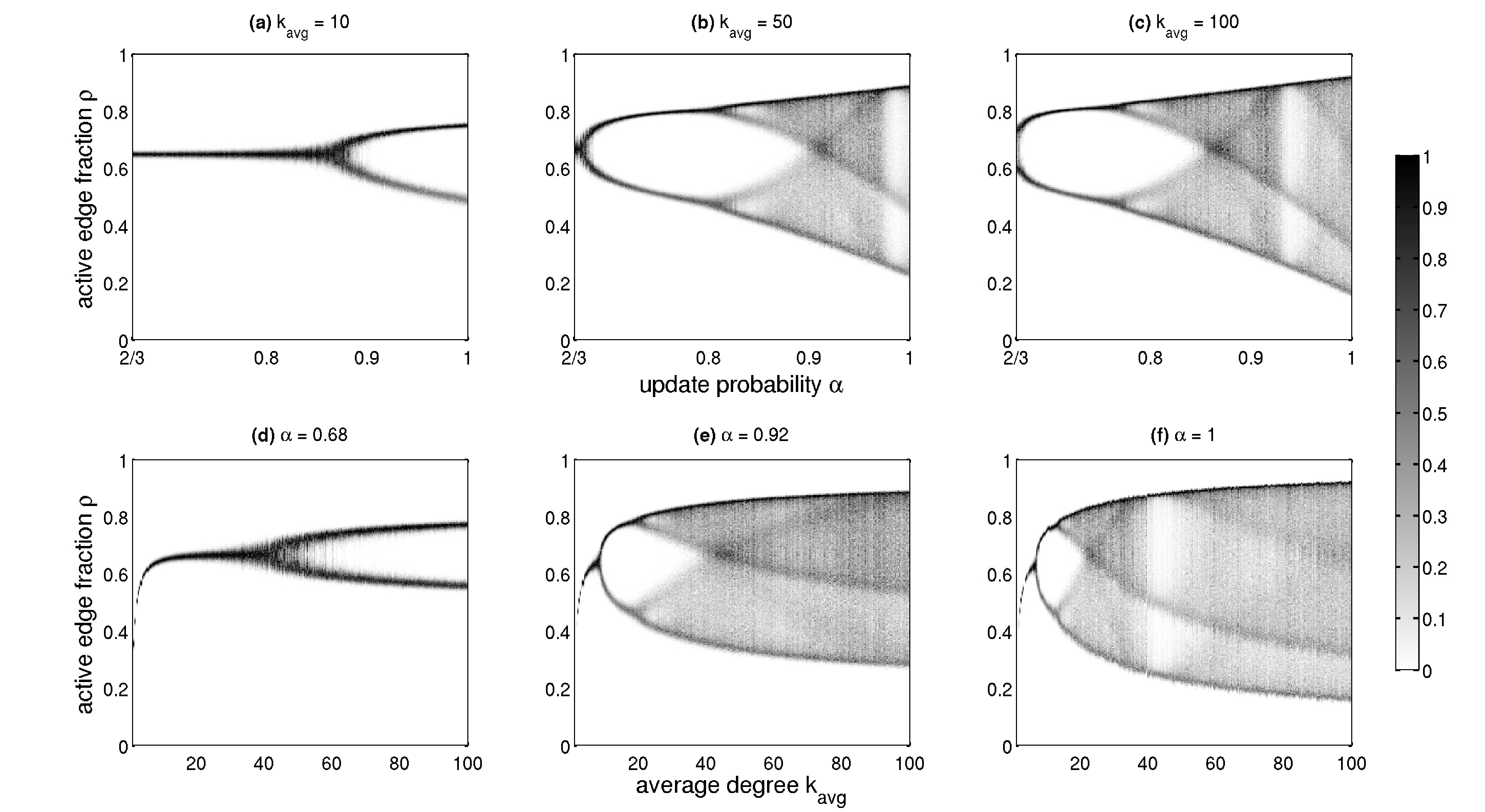}
  \caption[Bifurcation diagram from fully stochastic (P-R) simulations.]
  {
    Bifurcation diagram from fully stochastic (P-R) simulations.
    The same parameters were used as in
    Figure~\ref{fig:bifurcslice_eqn},
    which has the same structure only more blurred.
  }
  \label{fig:bifurcslice_sim}
\end{figure}

We explore the mean field dynamics by examining the limiting
behavior of the active edge fraction $\rho$
under the map $G(\rho; \avgdeg, \alpha)$.
We simulated the map dynamics
for a mesh of points in the $(\avgdeg, \alpha)$ plane.
We plotted the 3-dimensional 
(3-d; $N$-d denotes $N$-dimensional)
bifurcation structure of the mean field theory in Figure~\ref{fig:3dbifurc}.
We also made 2-d bifurcation plots for
fixed $\avgdeg$ and $\alpha$ slices through this volume, 
shown in
Figures~\ref{fig:bifurcslice_eqn} and~\ref{fig:bifurcslice_sim}.
In all cases, the invariant density of $\rho$
is normalized by its maximum for that $(\avgdeg, \alpha)$ pair
and indicated by the grayscale value.

The mean field map dynamics
exhibit period-doubling bifurcations
in both parameters $\avgdeg$ and $\alpha$. Visualizing the bifurcation
structure in 3-d (Figure~\ref{fig:3dbifurc}) shows 
interlacing period-doubling cascades in the two parameter dimensions.
These bifurcations are more clearly resolved when we take slices
of the volume for fixed parameter values.
The mean field theory (Figure~\ref{fig:bifurcslice_eqn})
closely matches the P-R simulations (Figure~\ref{fig:bifurcslice_sim}).
The first derivative 
${\partial G(\rho; \avgdeg, \alpha)}/{\partial \rho}
< {\partial \Phi(\rho; \alpha)}/{\partial \rho}$ 
for any finite $\avgdeg$, 
so the bifurcation point 
$\alpha = 2/3$ 
which we found for the dense map $\Phi$
is an upper bound for the first bifurcation point of $G$.
The actual location of the first bifurcation point depends on $\avgdeg$, 
but $\alpha =2/3$ becomes more accurate for higher $\avgdeg$ 
(it is an excellent approximation in Figures~\ref{fig:bifurcslice_eqn}c
and~\ref{fig:bifurcslice_sim}c, where $\avgdeg=100$). 
When $\alpha =1$, the first bifurcation point occurs at $\avgdeg \approx 7$.

The bifurcation diagram slices resemble each other and evidently
fall into the same universality class as the logistic map
\citep{feigenbaum1978a,feigenbaum1979a}.
This class contains all 1-d maps with a single,
locally-quadratic maximum. Due to the properties of the
Bernstein polynomials, $F_k(\rho; f)$ will universally
have such a quadratic maximum for any concave, continuous $f$ 
\citep{phillips2003a}. So this will also be 
true for $g(\rho; \avgdeg, f)$ with $\avgdeg$ finite, and we see that $\avgdeg$
partially determines the amplitude of that maximum in 
Figure~\ref{fig:tentmap}. 
Thus $\avgdeg$ acts as a bifurcation parameter. 
The parameter $\alpha$ tunes between 
$G(\rho; \avgdeg, 1) = g(\rho; \avgdeg, f)$ 
and 
$G(\rho; \avgdeg, 0) = \rho$, 
so it has a similar effect.
Note that the tent map $f$ and the dense limit map
$\Phi$ are kinked at their maxima, so their bifurcation
diagrams are qualitatively different from those of the mean field.
The network, by constraining the interactions among the population,
causes the mean field behavior to fall
into a different universality class than
the response function map.

\section{Conclusions}
\label{sec:conclusions}

We constructed the on-off threshold model as a simple model for
social contagion resulting from limited imitation. 
We see that including an 
aversion to total conformity results in more
complicated, even chaotic dynamics.
This model also allows us to study the effects 
of differing amounts of fixedness
in the social network and individual response functions,
and we developed a detailed mean field theory
which is exact for random mixing versions of the model.
Finally, we applied the theory to a specific case, 
where the network is Poisson and the response functions
average to the tent map. 

The model exhibits rich mathematical behavior.
The deterministic case, which we have barely touched on
here, merits further study. In particular, we would like to 
characterize the distribution of periodic sinks, how the 
collapse time scales with system size, and how similar
the transient dynamics are to the mean field dynamics.

Furthermore, the model should be tested on realistic networks.
These could include power law or small world random graphs,
or real social networks gleaned from data.
In a manner similar to \citet{melnik2011a}, 
one could evaluate the accuracy of the mean field theory 
for real networks.

Finally, the ultimate validation of this model would emerge from a
better understanding of social dynamics themselves. Characterization
of people's true response functions is therefore critical
\citep[some work has gone in this direction; see][]
{centola2010a, centola2011a, romero2011a, ugander2012a}. 
Comparison of model output to large data sets, such
as observational data from social media or online experiments,
is an area for further experimentation.
This might lead to more complicated context- and history-dependent models.
As we collect more data and refine experiments, 
the eventual goal of quantifiably predicting human behaviors, 
including fashions and trends,
seems achievable.
\end{doublespace}

\newpage
\cleardoublepage
\addcontentsline{toc}{section}{Bibliography}
\bibliographystyle{apalike}
\bibliography{library}

\cleardoublepage
\begin{appendices}
  \setcounter{figure}{1}
  \setcounter{table}{1}
  \section{Proof of Lemma 1}

\label{ax:limitthm}

\begin{mythm}
  \label{thm:1}
  For $k \geq 1$, let $f_k$ be continuous real-valued 
  functions on a compact domain $X$
  with $f_k \to f$ uniformly.
  Let $p_k$ be a probability mass function on $\mathbb{Z}^+$
  parametrized by its mean $\mu$ and with standard deviation
  $\sigma(\mu)$, assumed to be $o(\mu)$.
  Then,
  \begin{equation*}
    \label{eq:Top_lim}
    \lim_{\mu \to \infty} \left( \sum_{k=0}^\infty p_k f_k \right) = f .
  \end{equation*}
\end{mythm}

\begin{proof}
  Suppose $0 \leq a < 1$ and let $K = \lfloor \mu-\mu^a \rfloor$. Then,
  \begin{equation}
    g = \sum_{k=0}^{\infty} p_k f_k 
    = \sum_{k=0}^{K} p_k f_k + \sum_{k=K+1}^\infty p_k f_k.
    \label{eq:splitsum}
  \end{equation}
  Since $f_k \to f$ uniformly as $k \to \infty$,
  for any $\epsilon > 0$ we can choose $\mu$
  large enough that
  \begin{equation}
    \label{eq:epsilon}
    |f_k(x) - f(x)| < \epsilon
  \end{equation} 
  for all $k > K$ and all $x \in X$. 
  Without loss of generality, assume that $|f_k| \leq 1$ for all $k$. Then,
  \begin{equation*}
    |g-f| \leq \left( \frac{\sigma}{\mu^a} \right)^2 + \epsilon .
  \end{equation*}
  The $\sigma / \mu^a$ term is a consequence of the 
  Chebyshev inequality \citep{bollobas2001a}
  applied to the first sum in \eqref{eq:splitsum}.
  Since $\sigma$ grows sublinearly in $\mu$, this
  term vanishes for some $0 \leq a <1$
  when we take the limit $\mu \to \infty$.
  The $\epsilon$ term comes from using \eqref{eq:epsilon}
  in the second sum in \eqref{eq:splitsum}, 
  and it can be made arbitrarily small.
\end{proof}

\section{Online material}

To better explore the 3-d mean field bifurcation structure,
we created movies of the the $\avgdeg$ and $\alpha$ slices
as the parameters are dialed. The videos are available
at \\
\url{http://www.uvm.edu/\~kharris/on-off-threshold/bifurc\_movies.zip}.\\
Also, a VTK file with the 3-d bifurcation data, 
viewable in Paraview is at\\
\url{http://www.uvm.edu/\~kharris/on-off-threshold/volume_normalized.vtk.zip}.

Videos of the individual-node dynamics for small networks in 
the D-F and P-F cases are shown for some parameters
which produce interesting behavior. These are available at\\
\url{http://www.uvm.edu/\~kharris/on-off-threshold/graph\_movies.zip}.

The D-F, P-F, and P-R cases were implemented in Python.
The code is available at\\
\url{http://www.uvm.edu/\~kharris/on-off-threshold/code.zip}.
\end{appendices}

\end{document}